\documentclass[10pt,article]{IEEEtran}
\usepackage{srcltx,pdfsync}
\usepackage{amsmath}
\usepackage{amsfonts}
\usepackage{nopageno}
\usepackage{amssymb}
\usepackage{graphicx}
\usepackage{balance}
\usepackage{epstopdf}
\usepackage{cite}
\usepackage{url,ifthen}
\usepackage{caption}
\usepackage{subcaption}
\usepackage{mathrsfs}
\usepackage{amsthm}
\usepackage[dvipsnames,usenames]{color}
\usepackage[colorlinks, linkcolor=Blue, filecolor =Red,
citecolor=RoyalPurple]{hyperref}
\usepackage[lmargin=0.76in,rmargin=0.76in,tmargin=0.73in,bmargin=0.79in]{geometry}

\newcommand{\version}{arxiv}
\newtheorem{definition}{Definition}
\newtheorem{lem}{Lemma}
\newtheorem{prop}{Proposition}
\newtheorem{assumption}{Assumption}
\newtheorem*{proof-non}{Proof}
\newtheorem{comment}{Comment}
\newtheorem{cor}{Corollary}

\def\textAgg{\text{\scriptsize agg}}

\def\eqdef{\ensuremath{:=}}
\def\numBins{\ensuremath{n_{\text{bin}}}}
\def\numLoads{\ensuremath{n}}
\def\numLoadsBins{\ensuremath{n_{\ell}}}
\def\sumOfPSD{\ensuremath{\Sigma_{\bar{P}}}}
\def\sumOfPSDAgg{\ensuremath{ \Sigma^{\textAgg} }}
\def\sumOfPSDEll{\ensuremath{\Sigma_{\bar{P}^\ell}}}
\def\netDem{\text{\scriptsize ND}}
\def\loadDem{\text{\scriptsize BA}}
\def\PSD{\text{SD}}
\def\COP{\text{\tiny COP}}
\def\INT{\text{\scriptsize int}}
\def\SOC{\text{\scriptsize SoC}}
\def\textSmall{\text{\scriptsize Small}}
\def\textLarge{\text{\scriptsize Large}}
\def\textLow{\text{\scriptsize Low}}
\def\textHigh{\text{\scriptsize High}}

\newlength{\noteWidth}
\long\def\notes#1{\ifinner
           {\footnotesize #1}
           \else
           \marginpar{\parbox[t]{\noteWidth}{\raggedright\footnotesize #1}}
       \fi\typeout{#1}}
\def\notes#1{\typeout{read notes: #1}}  %uncomment for final version
\def\pb#1{\notes{pb: {\color{red}{#1}}      }}

\newboolean{showArxiv}
\newboolean{showArxivAlt}

\ifthenelse{\equal{\version}{journal}}
{
	\setboolean{showArxiv}{false} 
	\setboolean{showArxivAlt}{true}
}
{
}

\ifthenelse{\equal{\version}{arxiv}}
{
	\setboolean{showArxiv}{true} 
	\setboolean{showArxivAlt}{false}
}
{
}

\usepackage{xparse}

\ExplSyntaxOn
\cs_new:Nn \pb_prop_gset_bykeys:Nn
{
	\prop_clear:N \l__pb_temp_prop
	\keys_set:nn { pb/propbykey } { #2 }
	\prop_gset_eq:NN #1 \l__pb_temp_prop
}
\cs_generate_variant:Nn \pb_prop_gset_bykeys:Nn { c }

\keys_define:nn { pb/propbykey }
{
	unknown .code:n = \prop_put:Nxn \l__pb_temp_prop { \l_keys_key_tl } { #1 }
}

\cs_generate_variant:Nn \prop_put:Nnn { Nx }

\NewDocumentCommand{\setupcollaborator}{mm}
{% #1 = identifier string, #2 = set of key-value pairs
	\prop_new:c { g_collaborator_#1_prop }
	\pb_prop_gset_bykeys:cn { g_collaborator_#1_prop } { #2 }
}

\NewDocumentCommand{\selectcollaborator}{m}
{
	\prop_map_inline:cn { g_collaborator_#1_prop }
	{
		\tl_set:cn { ##1 } { ##2 }
	}
}
\ExplSyntaxOff

\setupcollaborator{GZlinux}
{% GZlinux is using his lab computer
	NAME=GZ on his lab computer,
	DiCEbibPATH=../../../../../DiCElab_bib,
	JBbibPATH=../../../../../bibs,
}
\setupcollaborator{Austin}
{% GZlinux is using his lab computer
	NAME=AC on his lab computer,
	DiCEbibPATH=../../../bibFiles/dicelab_bib,
	JBbibPATH=../../../bibFiles/bibs,
}
\setupcollaborator{PBmacbookair}
{% pb, macbook air or mac mini at work, 
	NAME=Prabir Macbookair,
	DiCEbibPATH=/Users/pbarooah/Dropbox/Dropbox/DiCElab_bib,
        JBbibPATH=/Users/pbarooah/Dropbox/Dropbox/bibs,
      }
\setupcollaborator{PBmac}
{% pb, macbook air or mac mini at work, 
	NAME=Prabir mac mini,
	DiCEbibPATH=/Users/pbarooah/Dropbox/DiCElab_bib,
    JBbibPATH=/Users/pbarooah/Dropbox/Dropbox/bibs,
}
\selectcollaborator{Austin}  

\begin{document}
\bstctlcite{IEEEexample:BSTcontrol} %needed so repeated authors aren't ---
\title{\vspace{6.4mm} Characterizing capacity of flexible loads for providing grid support}
\author{
	\IEEEauthorblockN{Austin R. Coffman\IEEEauthorrefmark{1}, Zhong Guo\IEEEauthorrefmark{1}$^,$\IEEEauthorrefmark{2}, and Prabir Barooah\IEEEauthorrefmark{1}}  \vspace{-0.7cm}
	
	\thanks{\IEEEauthorrefmark{1} University of Florida}
	\thanks{\IEEEauthorrefmark{2} corresponding author, email: zhong.guo@ufl.edu.}
	\thanks{ARC, ZG, and PB are with the Dept. of Mechanical and Aerospace Engineering, University of Florida, Gainesville, FL 32601, USA. The research reported here has been partially supported by the NSF through award 1646229 (CPS-ECCS).}
}

\maketitle
\thispagestyle{empty}
\begin{abstract}
Flexible loads are a resource for the Balancing Authority (BA) of the future to aid in the balance of supply and demand in the power grid. Consequently, it is of interest for a BA to know how much flexibility a collection of loads has, so to successfully incorporate flexible loads into grid level resource allocation. Loads' flexibility is limited by all their Quality of Service (QoS) requirements.  In this work we present a characterization of capacity for a collection of flexible loads. This characterization is in terms of the Power Spectral Density (PSD) of the  reference signal. Two advantages of our characterization are: (i) it easily allows for a BA to use the characterization for resource allocation of flexible loads and (ii) it allows for precise definitions of the power and energy capacity for a collection of flexible loads.
\end{abstract}
\section{Introduction}
The inherent variability in renewable generation sources such as solar and wind is a challenge for the power grid operators to balance demand and generation. Ramp rate constraints prevent conventional generation from handling this mismatch between demand and generation completely, while grid level storage from batteries is expensive. Thus a new resource is being investigated to help fill the mismatch: flexible loads. Flexible loads have the ability to vary power consumption over a baseline level without violating their Quality of Service (QoS). The baseline power consumption is the power consumed without grid interference. The requested amount from the grid authority, to deviate from baseline, is the \emph{reference signal}. The tracking of a zero-mean reference signal guises, in the eyes of the grid operator, flexible loads as batteries providing storage services. This battery-like behavior of flexible loads is often referred as Virtual Energy Storage (VES)~\cite{bar:2019}. VES  from flexible loads can be less expensive than energy storage from batteries~\cite{cammardella2018energy}. Some examples of flexible loads include residential air conditioners~\cite{CoffmanStudyHPB:2018}, water heaters~\cite{liu2019trajectory}, refrigerators, commercial HVAC systems~\cite{haokowlinbarmey:2013}, pumps for irrigation~\cite{AghajFarm:2019}, pool cleaning~\cite{chebusmey17a} or heating~\cite{LeeGridJESBC:2020}.

If the grid operator expects the flexible loads to track the reference signal accurately, then the reference must not cause the loads to violate their QoS. From the viewpoint of the grid operator, flexible loads not tracking a reference makes them appear unreliable. From the viewpoint of the load, reference signals that continually require QoS violation provide incentive for loads to stop providing VES. In either case, avoidance of the above scenarios is paramount to the long term success of VES. That is, reference signals must be designed to respect the \emph{capacity} of the collection of flexible loads.

%We say a collection of reference signals 
%The permissible set of reference signals conceptually defines the \emph{capacity} of the collection. %The ability to extract references signals from, and knowledge of, the capacity set would then allow for a grid operator to include flexible loads in autonomous grid support programs. 
Informally, the capacity represents limitations in ensemble behavior (e.g., the ability to track a reference signal) due to QoS requirements at the individual load. Consequently, a key step in determining the capacity is relating the QoS requirements to requirements for the reference signal. 
Unfortunately, this is not a straight forward task and many varying approaches are present in the current literature~\cite{hao_aggregate:2015,haowuliayan:2017,coffman2019aggregate:arxivACC,buildFlex:yin:2016,HaoKalsi:CDC:15,HughesPoolla:HICCS:15}. The most popular approach is to develop ensemble level necessary conditions~\cite{coffman2019aggregate:arxivACC,hao_aggregate:2015}; reference signals that satisfy these conditions ensure the ability of all loads in the collection to satisfy QoS while tracking the reference. Other approaches include geometry based characterizations~\cite{KunduKalsi:PSCC:18} and  characterizations through distributed optimization~\cite{Lin_ACC2018}\pb{What is a ``load centric'' characterization? And how is it different from ours?}. 

%simulation~\cite{} (\rd{need citation, I will find one later.}) and

A limitation of the currently available capacity results is that a balancing authority requires a specific reference signal to perform resource allocation; only post-facto checking if the BA's needs are within the capacity of the resource is possible. That is, these capacity characterizations cannot be used for long-term planning. The BA's need are exogenous, so a useful capacity characterization should allow the BA to determine the portion of its needs that can be provided by a specific class and number of loads ahead of time.
        
%Some limitations of the above mentioned works on capacity characterization are: (i) the capacity is characterized in terms of decision variables the grid operator cannot control, (ii) obtaining a realization of the capacity requires the use of complicated compute-heavy algorithms, or (iii) the characterization is based on constraints of the reference signal rather than constraints on statistical properties of the reference signal.         

%Characterizations in point (iii) require the reference signal in order to determine if it is within the capacity of the collection of flexible loads. Thus a grid authority requires a specific reference signal to perform resource allocation, which is adequate for hour or day ahead planning but inadequate for longer term planning. 

Contrarily, if one develops constraints on the statistics of the reference signal, then a notion of capacity that is also useful for planning can be developed. For instance, consider Figure~\ref{fig:Resource_Power} where the Power Spectral Density (PSD) of a grid's net demand is allocated to resources (precise definition of PSDs is provided later). This frequency based allocation does not require knowledge of a specific reference signal, but only of the statistics of the reference signal. To avoid possible confusion with electrical power (in Watt) we use Spectral Density (\PSD) instead of Power Spectral Density (PSD) in the rest of the paper.    

%Elaborating on point (iii), since many of the past works develop constraint on the reference signal for a collection of loads, the grid operator cannot answer questions such as: how much of a certain class of loads will be needed to balance supply and demand? 

In fact, it was introduced in~\cite{barbusmey:2015} that the QoS requirements of flexible loads can be characterized as constraints in the frequency domain. That is, while primarily an illustrative example, it is possible to quantitatively develop the regions shown in Figure~\ref{fig:Resource_Power}. 

%That is, it is possible to use the QoS requirements to develop constraints on the frequency content of the reference signal for the given class of flexible loads. 
%It is worth mentioning that in faster time scales it is possible for flexible loads to directly measure grid frequency to provide VES~\cite{zhatoplow:2012}. Although, the grid operator still requires knowledge of the capacity of the flexible loads, so that it can schedule other resources correctly. The point being, a characterization of capacity is valuable to the grid operator irrespective of the control architecture for the flexible loads. 

\begin{figure} [t]
	\centering
	\includegraphics[width=1\columnwidth]{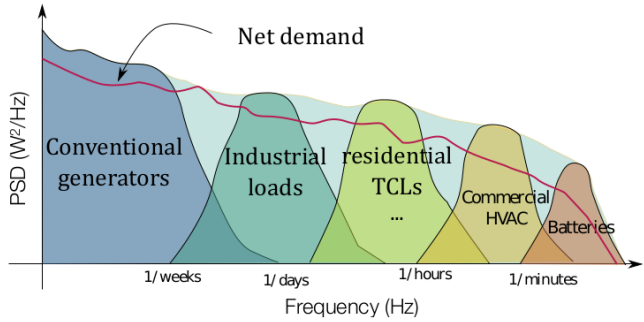}
	\caption{An example spectral allocation of resources to meet the grids needs.}
	\label{fig:Resource_Power}
\end{figure}

In addition to~\cite{barbusmey:2015,coffmanSpectral_ACC:2020}, there are many works that advocate for the specification of loads' abilities~\cite{krieger:Thesis:2013,haokowlinbarmey:2013} and for resource allocation spectrally~\cite{APT:2007}. The results of real world VES experiments also suggest that specifying the spectral content of a reference signal is a feasible way to encapsulate the limitations in flexibility of a load~\cite{linbarmeymid:2015}. A simplification of this concept is also widely used in today's power grid; ancillary services are classified by their response times and ramp rates~\cite{kirby2007ancillary}.

%Motivated by (i) the ability to quantify the QoS requirements spectrally and (ii) a prevalence of experiments that demonstrate the validity of frequency domain characterization, we formalize further the concept in~\cite{barbusmey:2015} for collections of flexible loads. 

Motivated by the advantages of working in the frequency domain we extend the work in~\cite{barbusmey:2015,coffmanSpectral_ACC:2020} and characterize the capacity through the \PSD\ of the reference signal, for an entire collection of flexible loads, based on the QoS requirements of each flexible load. The QoS constraints considered here are quite general and encapsulate operating constraints for: (i) Commercial HVAC systems, (ii) batteries, and (iii) Thermostatically Controlled Loads (TCLs). The contribution over past work is threefold: (i) we characterize capacity through constraints on the statistics of the reference signal, rather than the reference signal itself, (ii) our characterization of capacity allows for a BA to easily perform resource allocation and (iii) our characterization allows for precise definitions of the power and energy capacity for a collection of flexible loads.

We corroborate the advantages of our capacity characterization through numerical experiments. A convex optimization problem is setup that `projects' the needs of the grid authority onto the set of feasible \PSD's. The needs of the grid is quantified as the \PSD\ of net-demand, as seen in Figure~\ref{fig:Resource_Power}. In each simulation scenario, we determine the power and energy capacity of a collection of flexible loads based on the obtained \PSD. 

The capacity characterization presented is for continuously varying loads, that is loads that can vary power consumption freely within an interval. The capacity characterization can be extended to handle discrete loads (e.g., TCLs), by applying it to the work~\cite{coffman2019aggregate:arxivACC}. 

A preliminary version of this work is published in~\cite{coffmanSpectral_ACC:2020}, which dealt with a single load. In this work we extend the results to a collection of loads.

The paper proceeds as follows, Section~\ref{sec:indLoadModel} describes the QoS constraints of loads, Section~\ref{sec:indvidLoad} describes mathematical prerequisites and a characterization of individual load capacity, Section~\ref{sec:aggCap} describe the spectral characterization of the constraints for the ensemble. Section~\ref{sec:psdFeas} describes the method for determining how much of a grid's needs can be satisfied by flexible loads. Numerical results are provided in Section~\ref{sec:numExp}.

\section{QoS Constraints of Individuals} \label{sec:indLoadModel}
Denote by $P(t)$ the power consumption of a VES system at time $t$, and let $P^b(t)$ its  baseline demand. The demand deviation is $\tilde{P} \eqdef P(t) - P^b(t)$, and for the load to provide VES service by controlling the deviation $\tilde{P}(t)$ to track some reference signal. The first constraint is simply an actuator constraint:
\begin{align} \label{eq:QoSOne}
&\text{QoS-1:} \quad \left|\tilde{P}(t)\right| \leq c_1, \quad \forall \ t,
\end{align}
where the constant $c_1$, the maximum possible deviation of power consumption, depends on the rated power and the baseline demand. Second, define the demand increment $\tilde{P}_\delta(t) \eqdef \tilde{P}(t)-\tilde{P}(t-\delta)$, where $\delta>0$ is a predetermined (small) time interval. The second constraint is a ramping rate constraint:
\begin{align}
  \label{eq:QoS-increment}
  \text{QoS-2:} \quad \left|\tilde{P_\delta}(t)\right| \leq c_2, \quad &\forall \ t.
\end{align}
Third, define the additional energy use during any time interval of length $T$:
\begin{align} \label{eq:engDevDyn}
	\tilde{E}(t) = \int_{t-T}^{t}\tilde{P}(\sigma)d\sigma.
\end{align}
Since~\eqref{eq:engDevDyn} is a linear system it has representation with Laplace variable: $\tilde{E}(s) = G(s)\tilde{P}(s)$ for some BIBO stable transfer function $G$.
The third QoS constraint is that
\begin{align}\label{eq:QoS_energy}
\text{QoS-3:} \quad \left|\tilde{E}(t)\right| \leq c_3, \quad &\forall \ t.
\end{align}
The parameter $T$ in~\eqref{eq:engDevDyn} can represent the length of a billing period. Ensuring~\eqref{eq:QoS_energy} ensures that the energy consumed during a billing period close to the nominal energy consumed, although it is stronger than what is necessary. 

To define the fourth and last QoS constraint, we associate with the VES system a \emph{storage variable} $\theta$ that is related to the demand deviation through a stable linear time invariant system $H(s)$ as follows:
\begin{align} \label{eq:storage-LTI}
\tilde{\theta}(s) = H(s)\tilde{P}(s),
\end{align}
and impose the constraint
\begin{align}\label{eq:QoS_Thermenergy}
\text{QoS-4:} \quad \left|\tilde{\theta}(t)\right| \leq c_4, \quad &\forall \ t.
\end{align}
To understand the storage variable, imagine the VES system is a flexible HVAC system. A simple model of the indoor temperature of a building this HVAC system serves is:
\begin{align} \label{eq:thermDyn}
C\dot{\theta}(t) = \frac{1}{R}\left(\theta_0(t) - \theta(t)\right) + \dot{q}_{\INT}(t) + \eta_{\COP}P(t), 
\end{align}
where $C$ and $R$ are thermal capacitance and resistance, $\theta_0(t)$ is the ambient temperature, $\eta_{\COP}$ is the coefficient of performance, and $\dot{q}_{\INT}(t)$ is an internal disturbance. 

In general, the baseline power consumption for a HVAC system is the value $P^b(t)$ that keeps the internal temperature of the load at a fixed value $\bar{\theta}$, which for~\eqref{eq:thermDyn} is
\begin{align}
	P^b(t) = -\frac{\left(\theta_0(t)-\bar{\theta}\right)}{\eta_{\COP} R} - \frac{\dot{q}_{\INT}(t)}{\eta_{\COP}}.
\end{align} 
Since we are concerned with the flexibility in the load, we linearize~\eqref{eq:thermDyn} about the thermal setpoint $\bar{\theta}$ and the baseline power $P^b(t)$ yielding,
\begin{align}  \label{eq:devModel}
	\dot{\tilde{\theta}}(t) = -\gamma\tilde{\theta}(t) + \beta\tilde{P}(t), \quad \gamma = \frac{1}{RC}, \quad \beta = \frac{\eta_{\COP}}{C},
\end{align} 
where $\tilde{\theta}\triangleq \theta(t) - \bar{\theta}$ is the internal temperature deviation and $\tilde{P}$ is as defined at the beginning of this Section.

Taking the Laplace transform of~\eqref{eq:devModel}, we get $\tilde{\theta}(s) = \frac{\beta}{s+\gamma}\tilde{P}(s)$. So, for the HVAC system example, the storage variable is simply the indoor temperature deviation from the baseline value, and $H(s) = \frac{\beta}{s+\gamma}$. A more complex model of HVAC dynamics would lead to a higher order $H(s)$. The model~\eqref{eq:thermDyn}, while simplistic, has been shown to agree quite well with more realistic models for certain flexible loads~\cite{HUANG201958}.

Also, when the VES system is in fact a battery, $\tilde{\theta}(t)$ can be thought of a the amout of energy stored in the battery at time $t$, i.e., $\tilde{\theta} = E_0x_{\SOC}(t)$, where $E_0$ is the nominal energy capacity in kWh and $x_{\SOC}$ is the state of charge. A simple dynamic model of this variable is $\dot{\tilde{\theta}}(t) = -\alpha \tilde{\theta}(t)+\tilde{P}(t)$, where  $-\alpha \tilde{\theta}(t)$ accounts for the leakage, self degradation, and non-unity round trip efficiency of the battery. In this case $H(s) = \frac{1}{s+\alpha}$

The four constraints QoS 1-4, with parameters $Q^i \triangleq \left(\{c_i\}_{i=1}^4, T, \delta\right)^i$ specify the QoS set for the VES system.  The concern now is, what is a feasible power deviation signal $\tilde{P}(t)$?   

\section{Individual Load Capacity} \label{sec:indvidLoad}

\subsection{Stochastic Setting}
To develop our capacity characterization, we switch from a deterministic to a stochastic setting. In this setting we model the power deviation, $\tilde{P}$ as a continuous time stochastic process. The mean and autocorrelation function for $\tilde{P}$ are,
\begin{align}
	\mu_{\tilde{P}}(t) &\triangleq \mathbf{E}[\tilde{P}(t)], \quad \forall \ t, \\ 
	R_{\tilde{P}}(s,t) &\triangleq \mathbf{E}[\tilde{P}(s)\tilde{P}(t)], \ \forall \ s,t,
\end{align}
where $\mathbf{E}[\cdot]$ denotes mathematical expectation. We make the following assumption about the stochastic process $\tilde{P}$.
\begin{assumption} \label{ass:WSSprocess}
	The stochastic process $\tilde{P}$ is wide sense stationary (WSS) with mean function $\mu_{\tilde{P}} = 0$ for all $t$.
\end{assumption}

Since $\tilde{P}$ is expressed as the difference of the power consumption from a base value it is intuitive to set the expectation of this to zero. Otherwise, loads are not providing storage services. Furthermore, WSS requires the variance and mean of the process $\tilde{P}$ to be time invariant and the autocorrelation function to be a function of $\tau = s-t$. We denote the time invariant variance as $\sigma^2_{\tilde{P}}$; the time invariant mean is already specified in the assumption.

In addition, for a continuous time WSS stochastic process $\{X(t)\}$ we have, through the Fourier transform, an alternative expression for the autocorrelation function~\cite{hajek2015random},
\begin{align}\label{eq:WKT}
&R_{X}(\tau) = \int_{-\infty}^{\infty} S_{X}(\omega)  e^{(j\omega\tau)}d\frac{\omega}{2\pi}, \ \text{and} \\ \label{eq:WKTTwo}
&S_{X}(\omega) = \int_{-\infty}^{\infty} R_{X}(\tau)  e^{(-j\omega\tau)}d\tau,
\end{align}
where $S_{X}(\omega)$ is the (power) Spectral Density (\PSD) of $X$, $\omega \in \mathbb{R}$ is the frequency variable, and $j$ is the imaginary unit. The above is based on the general definition for (power) \PSD\ of the signal $X$,
\begin{align} \label{eq:altDefPSD}
	S_{X}(\omega) \triangleq \lim_{T \to \infty}\frac{1}{T}\mathbf{E}\bigg[\bigg|\int_{0}^{T}X(t)e^{-j\omega t}\bigg|^2\bigg]
\end{align}
 the equivalence of definitions~\eqref{eq:altDefPSD} and~\eqref{eq:WKTTwo} for a WSS process is the Wiener-Khinchin theorem. Letting $\tau = 0$ in~\eqref{eq:WKT} results in,
\begin{align}\label{eq:PSD_VAR}
R_{X}(0) = \int_{-\infty}^{\infty} S_{X}(\omega) d\frac{\omega}{2\pi} ,
\end{align}
where, if the mean of $X$ is zero, $R_X(0)$ is the variance of the process $X$. We introduce the following proposition from~\cite{hajek2015random} that will be useful in the developments to follow.

\begin{prop} \label{prop:hajekTwo}
	\cite{hajek2015random} Let X be a WSS stochastic process and input to the linear time invariant BIBO stable system $G(s)$ with output $Y$, then $Y$ is WSS, $X$ and $Y$ are jointly WSS, and
		\begin{align} \nonumber
		\text{(i)} \quad &\mathbf{E}[Y] = G\left(j\omega\right)\Big|_{\omega=0}\mathbf{E}[X], \\ \nonumber
		\text{(ii)} \quad &S_Y(\omega) = \left|G(j\omega)\right|^2S_X(\omega),
		\end{align}
		where $S_X$ is the \PSD\ of $X$, $S_Y$ is the \PSD\ of $Y$, and $G(j\omega)$ is the frequency response of $G(s)$.
\end{prop}
Furthermore, the Chebyshev inequality for a random variable $X$, will be useful:
\begin{align} \label{eq:chebIneq}
\mathcal{P}\big(\left|X - \mu_{\tiny X}\right| \geq k\big) \leq \frac{\sigma^2_X}{k^2}, \quad \forall \ k>0,
\end{align}
where $\mathcal{P}(\cdot)$ denotes probability.

\ifx 0
\begin{comment}[The WSS assumption]
The assumption that $\tilde{P}$ is WSS may appear arbitrary, and limiting in application of flexible demand in grid support. There non-stationary variations in the net demand due to diurnal, weekly, and seasonal trends and demand flexibility is needed to help reduce the net demand. However, since most flexible loads only have flexibility in the time scale of few hours, we expect the deviation $\tilde{P}(t)$ requested to be at faster time scale. The lower frequency components of the net demand will be met by traditional genertors. The remaining component, which flexible loads are expected to track, does not have these non-stationary trends. 
\end{comment}
\fi

\subsection{Inequality Constraints: Spectral Characterization} \label{sec:ineqCons}
The QoS constraints QoS~1-4 are characterized probabilistically in the following way. The inequalities in QoS~1-4 turn into probabilistic inequalities; the probability of the QoS constraint \textit{not} being met is required to be small:
\begin{align} %\label{eq:relexEqCons}
%&\mathbf{E}[\tilde{P}(t)] = 0, \quad \forall \ t \in \mathbf{T},  \\  
\label{eq:con_P_k}
&\mathcal{P}\left(\left|{\tilde{P}(t)}\right|\geq c_1\right) \leq \epsilon_1,  &\forall \ t,\\
\label{eq:con_P_k-1}
&\mathcal{P}\left(\left|\tilde{P}_\delta(t)\right|\geq c_2\right) \leq \epsilon_2,  &\forall \ t,\\
\label{eq:con_E_k}
&\mathcal{P}\left(\left|\tilde{E}(t)\right|\geq c_3\right) \leq \epsilon_3,  &\forall \ t,\\
\label{eq:con_T_k}
&\mathcal{P}\left(\left|\tilde{\theta}(t)\right|\geq c_4\right) \leq \epsilon_4,  &\forall \ t,
\end{align}
%where $\mathcal{P}\left(\cdot\right)$ denotes probability. %We have included the constraint~\eqref{eq:relexEqCons} to ensure that all signals have zero mean, which is discussed further in Section~\ref{sec:eqCon} 
The quantities $\{\epsilon_i\}_{i=1}^4$ set the tolerance level for satisfying the respective constraint and are chosen to be small. 
%In what follows the constraints~\eqref{eq:con_P_k}-\eqref{eq:con_E_k} will be expressed in terms of constraints on $S_{\tilde{P}}$, the PSD of the signal $\tilde{P}$.

In order to pose the inequality constraints~\eqref{eq:con_P_k}-\eqref{eq:con_T_k} in terms of $S_{\tilde{P}}$, two steps are taken. The first step is to utilize the Chebyshev inequality~\eqref{eq:chebIneq} to bound the probabilities in~\eqref{eq:con_P_k}-\eqref{eq:con_T_k} as a function of the variance of the given random variable. The second step is to then use the Wiener-Khinchin theorem~\eqref{eq:WKT} to express the variance as the integral of $S_{\tilde{P}}$.

\begin{lem} \label{lem:zeroMean}
	Let $\tilde{P}$ satisfy Assumption~\ref{ass:WSSprocess}, then for all $t$,
	\begin{align} \nonumber
	\mathbf{E}[\tilde{E}(t)] = 0, \ \mathbf{E}[\tilde{P}_{\delta}(t)] = 0, \ \text{and} \ \mathbf{E}[\tilde{\theta}(t)] = 0.
	\end{align}  
\end{lem}

\begin{proof}
	Apply the result of Proposition~\ref{prop:hajekTwo}-(i) for $\mathbf{E}[\tilde{E}(t)]$ and $\mathbf{E}[\tilde{\theta}(t)]$. The linearity of expectation suffices for $\mathbf{E}[\tilde{P}_{\delta}(t)]$.
\end{proof}
With the result in Lemma~\ref{lem:zeroMean} and Chebyshev's inequality~\eqref{eq:chebIneq} we formulate sufficient conditions for the inequality constraints~\eqref{eq:con_P_k}-\eqref{eq:con_E_k} as follows,
\begin{align} \label{eq:chebRowOne}
&\sigma_{\tilde{P}}^2 \leq {c_1^2}\epsilon_1,  &\sigma^2_{\tilde{P}_\delta} \leq {c_2^2}\epsilon_2,\\ \label{eq:chebRowTwo}
&\sigma^2_{\tilde{E}} \leq {c_3^2}\epsilon_3 
&\sigma^4_{\tilde{\theta}} \leq {c_4^2}\epsilon_4,
\end{align}
so that the probability of exceeding the inequality constraints~\eqref{eq:con_P_k}-\eqref{eq:con_T_k} will be less than the respective specified amount, $\{\epsilon_i\}_{i=1}^4$.
%Before transforming the LHS of~\eqref{eq:chebRowOne}-\eqref{eq:chebRowTwo} in terms of $S_{\tilde{P}}$ it is necessary to compute the variance $\sigma^2_{\tilde{P}_\delta}$ and $\sigma^2_{\tilde{E}}$, for which we partly rely on Proposition~\ref{prop:hajek}, the linearity of expectation, and the Wiener-Khinchin theorem,
The variance $\sigma^2_{\tilde{P}_\delta}$ is equivalently, 
\begin{align} 
&\sigma^2_{\tilde{P}_{\delta}} = \mathbf{E}\bigg[\left(\tilde{P}_\delta(t)\right)^2\bigg] =  2\left(R_{\tilde{P}}(0)-R_{\tilde{P}}(\delta)\right),  \label{eq:pDeltaVar} 
%&\sigma^2_{\tilde{E}}(T) =  \mathbf{E}\bigg[\left(\tilde{E}(T)\right)^2\bigg] = 2\int_{0}^{T}R_{\tilde{P}}(\tau)\left(T - \tau\right)d\tau   \nonumber \\
%&=2\int_{0}^{T}(T-\tau)\left(\int_{-\infty}^{\infty} S_{\tilde{P}}(\omega)  e^{(j\omega\tau)}d\frac{\omega}{2\pi}\right)d\tau \label{eq:eTildeVarOne}, \\
%&=\int_{-\infty}^{\infty}S_{\tilde{P}}(\omega)\left(\int_{0}^{T} (T-\tau)  \cos{(\omega\tau)}d\tau\right)d\frac{\omega}{\pi}. \label{eq:eTildeVar}, \\
\end{align}

Now applying the Wiener-Khinchin theorem~\eqref{eq:WKT} to the LHS of~\eqref{eq:chebRowOne}-\eqref{eq:chebRowTwo} we have,
\begin{align} 
%\sigma^2_{\tilde{P}} = 
&\int_{0}^{\infty} S_{\tilde{P}}(\omega) d\frac{\omega}{\pi} \leq {c_1^2}\epsilon_1, \label{eq:intPowerDev}\\ 
%
%\sigma^2_{\tilde{P}_\delta} = 
&\int_{0}^{\infty} S_{\tilde{P}}(\omega) \left(2-2\cos{(\omega\delta)}\right)d\frac{\omega}{\pi}  \leq \epsilon_2 c_2^2,\label{eq:invRateCon}\\
%
%\sigma^2_{\tilde{E}} = 
&\int_{0}^{\infty} S_{\tilde{E}}(\omega)d\omega =\int_{0}^{\infty} |G(j\omega)|^2S_{\tilde{P}}(\omega) d\frac{\omega}{\pi}\leq {c_4^2}\epsilon_4, \label{eq:intEngDev} \\
%
%\sigma^2_{\tilde{\theta}} = 
&\int_{0}^{\infty} S_{\tilde{\theta}}(\omega) d\frac{\omega}{\pi} = \int_{0}^{\infty} |H(j\omega)|^2 S_{\tilde{P}}(\omega) d\frac{\omega}{\pi} \leq {c_3^2}\epsilon_3 \label{eq:intTempDev}.
\end{align}
\begin{definition}(Individual load constraint set) \label{def:indConDef}
	The set of feasible \PSD s is,
	\begin{align} \nonumber
		\mathcal{S} = \{S_{\tilde{P}} \ : \ S_{\tilde{P}} \geq 0, \ \text{and} \ \eqref{eq:intPowerDev}-\eqref{eq:intEngDev}\}.
	\end{align}
\end{definition}
That is, if $\tilde{P}$ has \PSD\ in the set $\mathcal{S}$, then the loads power deviation can match $\tilde{P}$ while satisfying the constraints~\eqref{eq:con_P_k}-\eqref{eq:con_T_k}.
																
\subsection{Power and energy capacity}
Utilizing the concept of \PSD, the definition of power and energy capacity of a flexible load can be characterized in terms of the \PSD\ of the power deviation.
\begin{definition} \label{def:capDef}
	Let $\epsilon \in (0,1]$ and $\mathcal{S}_{\epsilon}$ be a set of \PSD s, then the power and energy capacity for a given \PSD\ $S\in\mathcal{S}_{\epsilon}$ are,
	\begin{align}\nonumber
		\textbf{Pow}(S) &= \sqrt{\frac{1}{\pi\epsilon}\int_{0}^{\infty}S(\omega)d\omega} \quad (W), \\ \nonumber
		\textbf{Eng}(S) &=\sqrt{\frac{1}{\pi\epsilon}\int_{0}^{\infty}\left|G(j\omega)\right|^2S(\omega)d\omega} \quad (Wh),
	\end{align}
	%where the square root is included to produce the correct units.
\end{definition}
The reason for these definitions is that  for a signal $\tilde{P}(t)$ with \PSD\ $S$, the following inequalities follow from these definitions upon applying~\eqref{eq:engDevDyn}:
\begin{align}
	\mathcal{P}\left(\left|\tilde{P}(t)\right| \geq \textbf{Pow}(S)\right) \leq \epsilon, \forall t\\
	\mathcal{P}\left(\left|\tilde{E}(t)\right| \geq \textbf{Eng}(S)\right) \leq \epsilon, \forall t.
\end{align}
The set $\mathcal{S}_{\epsilon}$ is given explicit dependence on $\epsilon$ to reflect the dependence of constraints~\eqref{eq:intPowerDev}-\eqref{eq:intEngDev} on $\epsilon$. %Practically $\mathcal{S}_{\epsilon}$ will be taken as $\mathcal{S}$ in Definition~\ref{def:indConDef}, however, Definition~\ref{def:capDef} is not limited to only this set.

\ifshowArxivAlt
This is an application of the Chebyshev inequality~\eqref{eq:chebIneq} with $k = \textbf{Pow}(S)$ ($\textbf{Eng}(S)$). We offer numerical experiments to further elucidate our capacity definitions in the extended version of this manuscript found on arxiv~\cite{bibid}.
\fi

\ifshowArxiv
For clarity, we estimate these probabilities for a given $\epsilon$ and \PSD\ $S \in \mathcal{S}_\epsilon$. We denote the estimated probabilities as $\hat{P}_i(\cdot)$ which are estimated by obtaining a time series of length $N$ with the desired \PSD\ and computing,
\begin{align} \nonumber
	\hat{\mathcal{P}}_i\left(\left|\tilde{P}\right| \geq \textbf{Pow}(S)\right) = \frac{1}{N}\sum_{k=1}^{N}\mathbf{1}\left(\left|\tilde{P}_k\right| \geq \textbf{Pow}(S)\right)
\end{align}
The results are repeated for $2\times 10^4$ trials, and are shown in Figures~\ref{fig:estProbPowBound}-\ref{fig:estProbEngBound} as a histogram, with values typically well below the respective thresholds. In addition, we plot two generated time series sample paths in Figures~\ref{fig:samplePathPow}-\ref{fig:samplePathEng}. The two sample paths are typically within the power and energy capacity bounds defined.

Thus for a \PSD\ $S$, the power and energy capacity defined in Definition~\ref{def:capDef} encompass maximum probabilistic limits, i.e., the capacity, for a time series that is associated with the \PSD\ $S$.

\begin{figure}
	\centering
	
	\includegraphics[width=1\columnwidth]{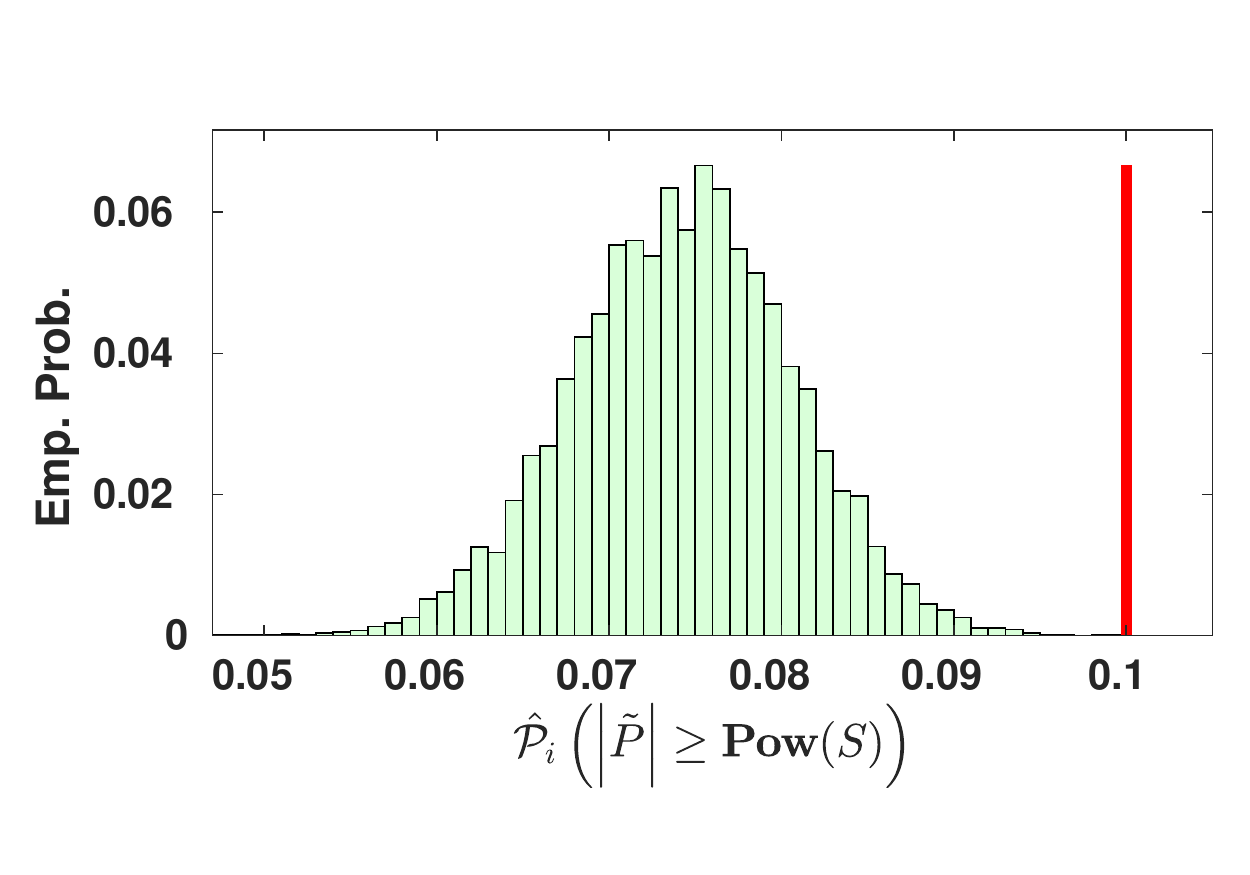}
	\caption{Empirically estimated probability bounds for the power capacity with $\epsilon_1 = 0.1$.}
	\label{fig:estProbPowBound}
	
	\centering
	\includegraphics[width=1\columnwidth]{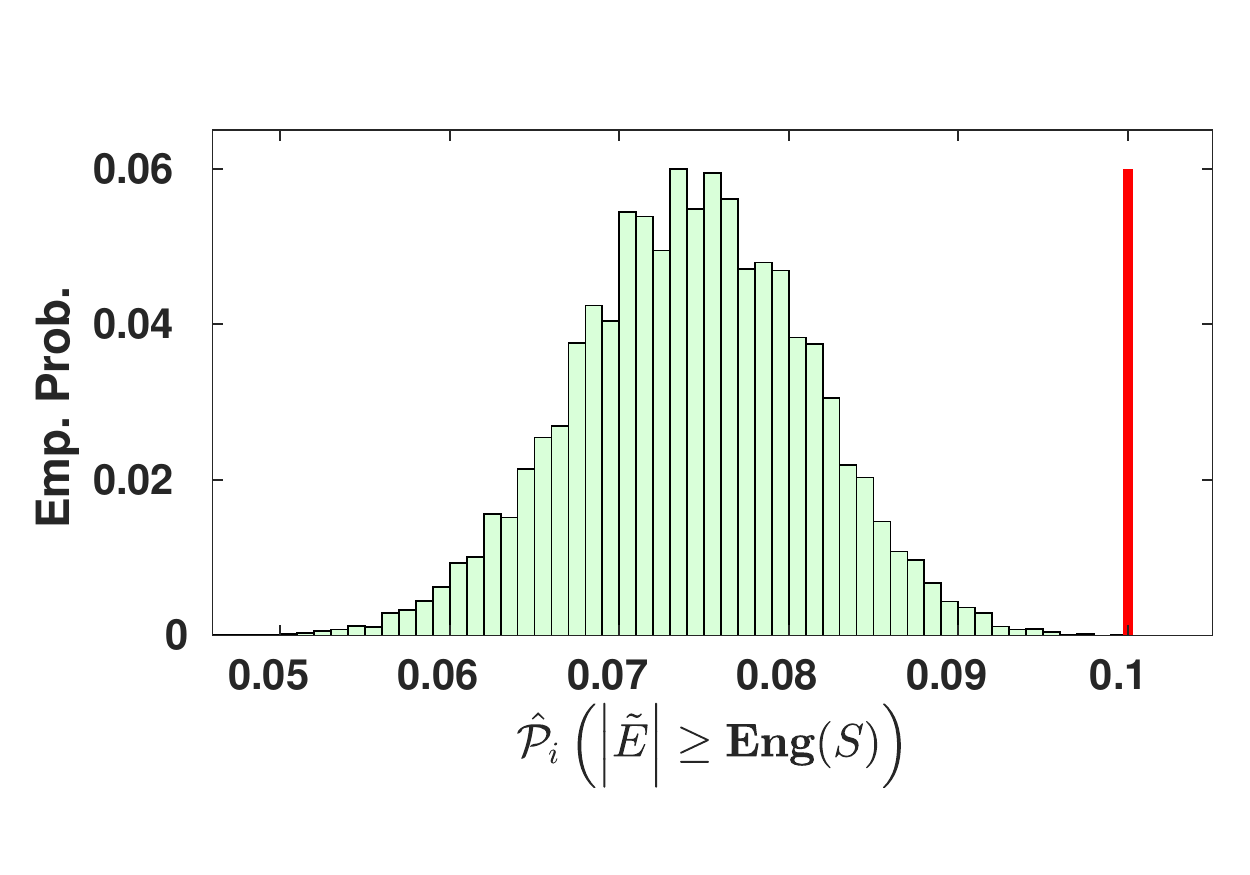}
	\caption{Empirically estimated probability bounds for the energy capacity with $\epsilon_4 = 0.1$.}
	\label{fig:estProbEngBound}

\end{figure}

\begin{figure}
	\centering
	\includegraphics[width=1\columnwidth]{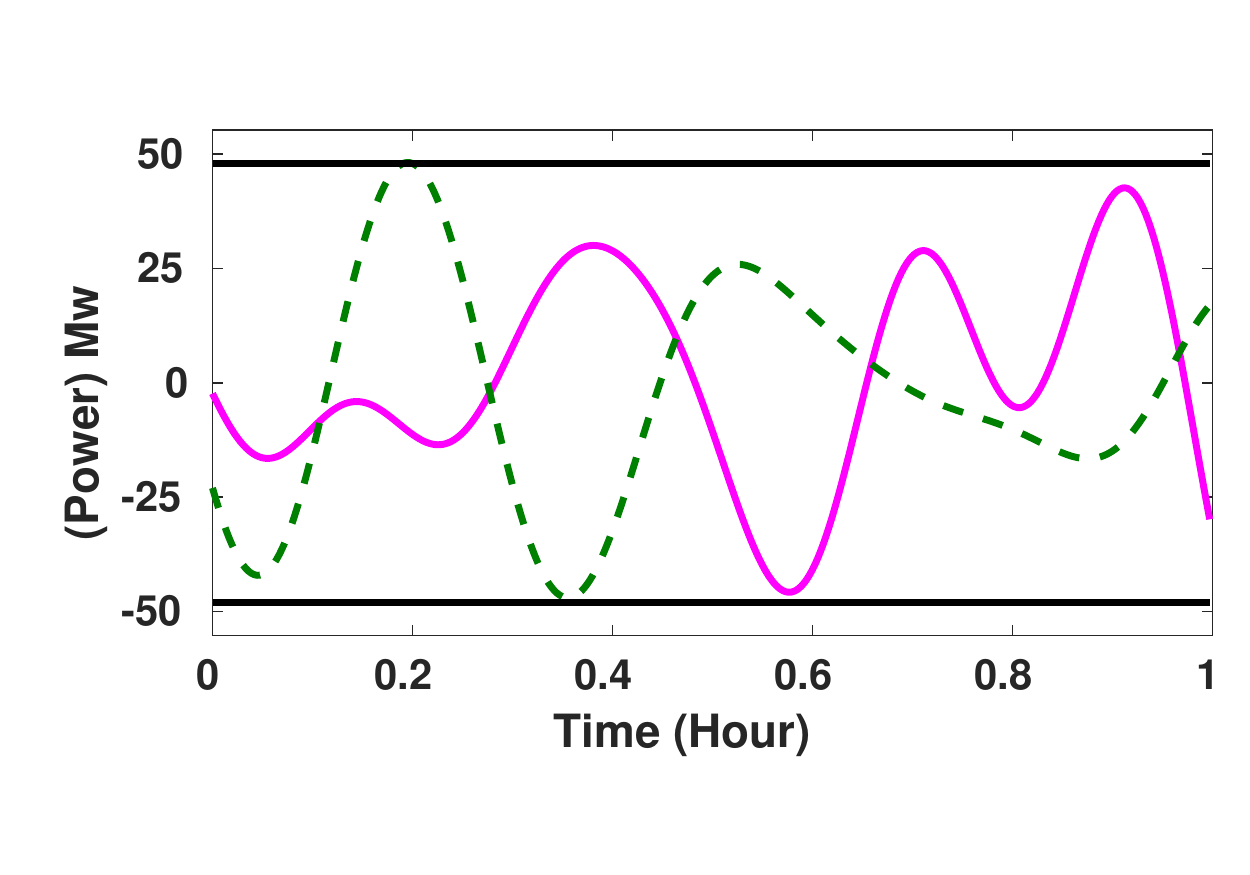}
	\caption{Two sample paths of a power consumption trajectory.}
	\label{fig:samplePathPow}
	
	\centering
	\includegraphics[width=1\columnwidth]{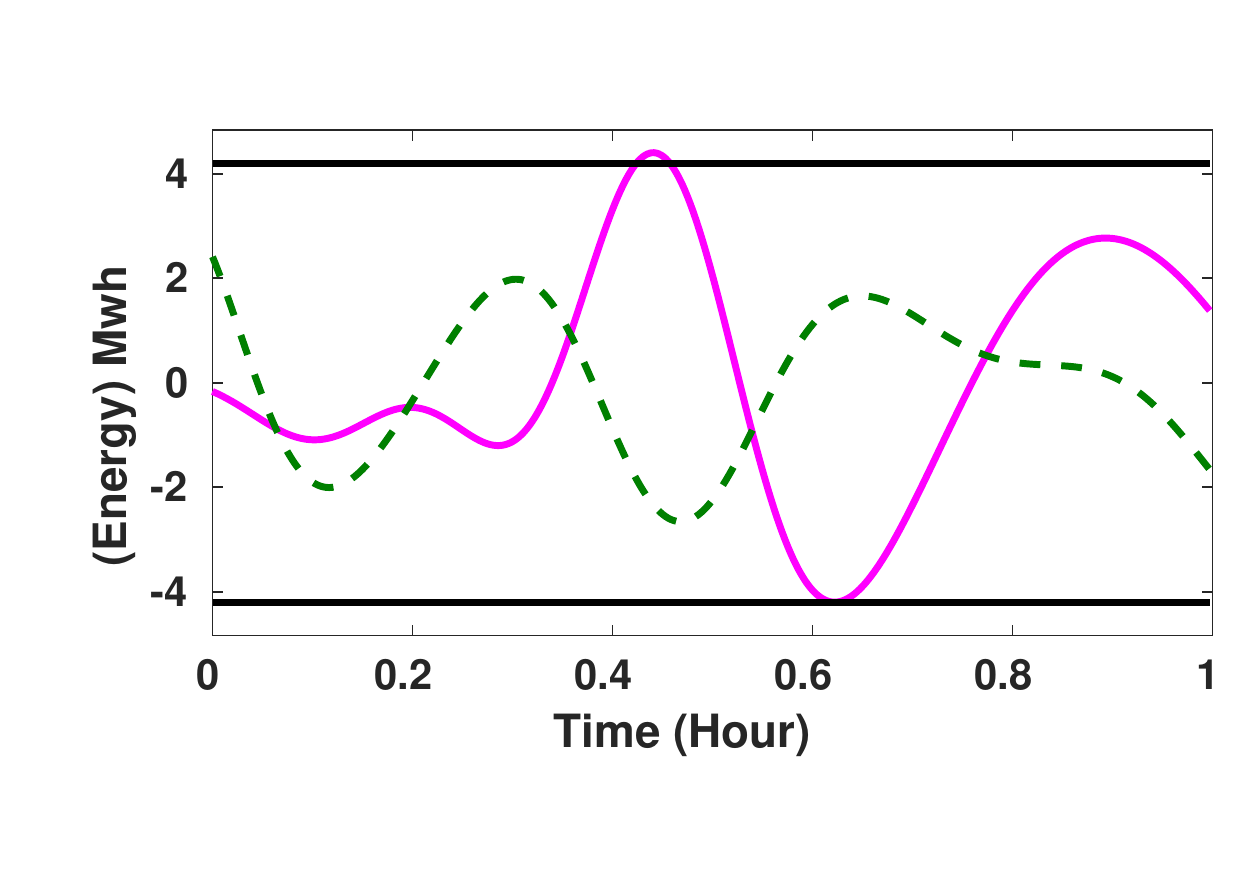}
	\caption{Two sample paths of an energy consumption trajectory.}
	\label{fig:samplePathEng}
	
\end{figure}

\begin{definition}
 	Let $\epsilon \in (0,1]$ and $S\in \mathcal{S}_\epsilon$ with power and energy capacity $\textbf{Pow}(S)$ and $\textbf{Eng}(S)$, respectively, then the un-used power and energy capacity is 
 	\begin{align}
	 	\mathcal{U}_{\textbf{Pow}}(S) &= c_1\left(1 - \frac{\textbf{Pow}(S)}{c_1}\right) \quad (W), \\
	 	\mathcal{U}_{\textbf{Eng}}(S) &= c_4\left(1 - \frac{\textbf{Eng}(S)}{c_4}\right) \quad (Wh),
 	\end{align}  
 	where $c_1$ and $c_4$ are as used in~\eqref{eq:QoSOne} and~\eqref{eq:QoS_energy}, respectively.
 \end{definition}
\fi
\section{Ensemble Capacity} \label{sec:aggCap}
Section~\ref{sec:indvidLoad} constructed constraints on the \PSD\ $S_{\tilde{P}}$ so to respect the capacity requirements for an individual load. In this section we do the same for \numLoads\ loads. This is achieved by defining an ``ensemble \PSD'' and utilizing the constraints on the individuals to develop constraints for the ensemble \PSD. First, the ensemble power deviation is defined as
\def\Psum{\bar{P}(t)}
\begin{align} \label{eq:ensemPowerDev}
\Psum\triangleq \sum_{i=1}^{\numLoads}\tilde{P}^i(t).
\end{align}
The ensemble \PSD\ is then the \PSD\ of $\bar{P}(t)$, that is,

\begin{align} \label{eq:ensemPSD}
	S_{\bar{P}}(\omega) = \int_{-\infty}^{\infty}R_{\bar{P}}(\tau)e^{-j\omega\tau}d\tau,
\end{align}
where $R_{\bar{P}}(\tau)$ is the autocorreleation function of $\bar{P}$. With the requirement that each $\tilde{P}^i(t)$ is jointly WSS, the ensemble power deviation~\eqref{eq:ensemPowerDev} is also WSS and the definition~\eqref{eq:ensemPSD} is valid.

\begin{figure}
	\centering
	\includegraphics[width=0.75\columnwidth]{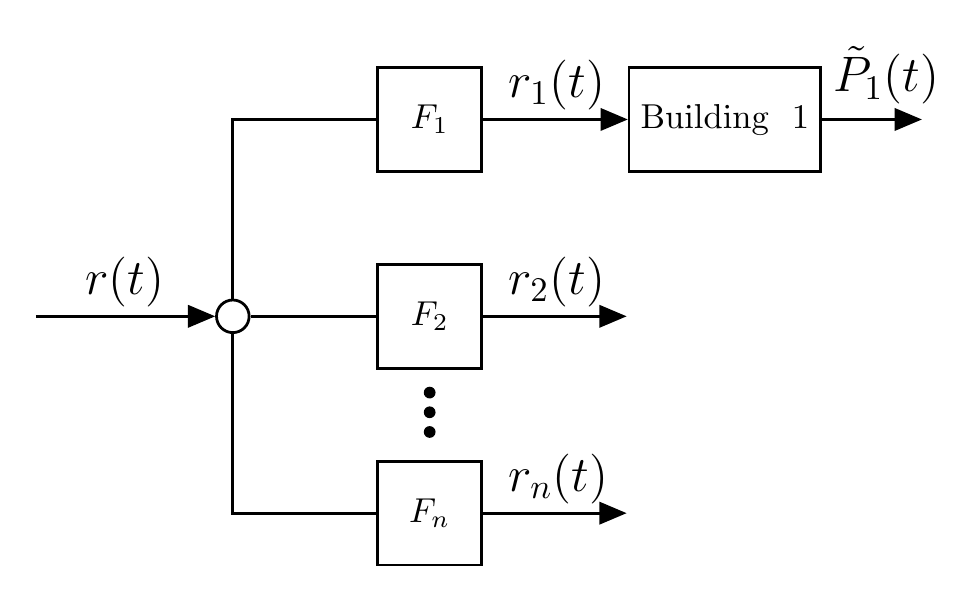}
	\caption{An example control architecture for flexible loads to assist the grid.}
	\label{fig:exampleControlArch}
\end{figure}

The \PSD\ of the sum, being the Fourier transform of the autocorrelation of the sum, depends on how the component signals $\tilde{P}^i$ are correlated to one another. In the limiting case when they are uncorrelated with one another, the autocorrelation of the sum is the sum of the autocorrelations since the cross correlations are 0. In that case the \PSD\ of the sum is also the sum of the \PSD s. Similarly, in the limiting cases of being perfectly anti-correlated or correlated, one can come up with precise relationship. However, the most interesting case from a practical point of view is when they signals are positively correlated but without having a perfect correlation of 1. To understand the reason, let us imagine how these signals $\tilde{P}^i$ will be generated. It is reasonable to assume a control architecture shows in Figure~\ref{fig:exampleControlArch}: a regulation signal from the balancing authority, $r(t)$, is split into $n$ distinct components, $r_i(t)$, $i=1,\dots,\numLoads$, so that $r_i(t)$ is within the capacity of load $i$, and some control system is used to ensure that $\tilde{P}^i$ tracks $r_i$. For instance, $r_i$ can be computed by band-pass filtering $r$ with a filter $F_i(j\omega)$. In a well designed system, these filters will ensure that the signals $r_i$ have sufficient positive correlation. Otherwise, each load - assuming they perfectly track their references - will be working against each other instead of working collaboratively to supply the total $r$. Although the design of the coordination architecture for the loads is beyond the scope of this work, we make the following assumption:
\begin{assumption}\label{ass:cross-corr-ass}
For any pair of loads $\ell,m$, $E[\tilde{P}_\ell(t)\tilde{P}_m(t+\tau)] =: R_{\ell,m}(\tau)\geq 0$ for every $\tau$.
\end{assumption}
\begin{lem} \label{lem:psdBounds}
	The ensemble \PSD\ of \numLoads\ loads, under Assumption~\ref{ass:cross-corr-ass}, satisfies
	\begin{align}\label{eq:PSDbounds-1}
		\sumOfPSD \leq S_{\bar{P}} \leq \numLoads \sumOfPSD,
	\end{align}
	where $\sumOfPSD \eqdef \sum_{i=1}^{\numLoads}S_{\tilde{P}^i}$.
\end{lem}
\begin{proof}
  Since the signals in the sum are jointly WSS, we have
  \begin{align}\label{eq:corr-of-sum}
  R_{\Psum,\Psum}(\tau) = \sum_{\ell=1}^\numLoads\sum_{m=1}^\numLoads R_{\ell,m}(\tau).
  \end{align}
  The cross correlation $R_{\ell,m}(\tau)$ when $\ell \neq m$  is non negative by Assumption~\ref{ass:cross-corr-ass} (the first one), $R_{\Psum,\Psum}(\tau) \geq \sum_{\ell=1}^n R_{\ell,\ell}(\tau)$, from which the lower bound follows due to linearity of the Fourier transform. The upper bound is proven by considering the definition~\eqref{eq:altDefPSD} for $S_{\bar{P}}(\omega)$. We have for all $\omega \in \mathbb{R}$, $S_{\bar{P}}(\omega) = $
  \begin{align}
  	 %&= \lim_{T \to \infty}\frac{1}{T}\mathbf{E}\bigg[\bigg|\int_{0}^{T}\bar{P}(t)e^{-j\omega t}\bigg|^2\bigg] \\
  	&\lim_{T \to \infty}\frac{1}{T}\mathbf{E}\bigg[\bigg|\sum_{i=1}^{\numLoads}\int_{0}^{T}\tilde{P}^i(t)e^{-j\omega t}\bigg|^2\bigg] \\
  	&\leq\lim_{T \to \infty}\frac{\numLoads}{T}\sum_{i=1}^{\numLoads}\mathbf{E}\bigg[\bigg|\int_{0}^{T}\tilde{P}^i(t)e^{-j\omega t}\bigg|^2\bigg] 
  	%=\numLoads\sum_{i=1}^{\numLoads}S_{\tilde{P}^i}(\omega) = \numLoads\sumOfPSD(\omega),
  	= \numLoads\sumOfPSD(\omega),
  \end{align}
  where the bound is from Jensen's inequality: $\left|\sum_{i=1}^{\numLoads}x_i\right|^2 \leq \numLoads\sum_{i=1}^{\numLoads}\left|x_i\right|^2$ since $\left|\cdot\right|^2$ is convex.
\end{proof}

\begin{cor} \label{cor:homEnsem}
	For a homogeneous collection of \numLoads\ loads in which $P_\ell(t) = P_m(t)$ for every $\ell,m =1,\dots,\numLoads$, the ensemble \PSD\ is
	\begin{align} \nonumber
		S_{\bar{P}} = \numLoads\sumOfPSD.
	\end{align}
\end{cor}
The proof follows from the same argument used in establishing the upper bound case of Lemma~\ref{lem:psdBounds}.

In addition, for a heterogeneous ensemble the lower bound in Lemma~\ref{lem:psdBounds} will be loose as \numLoads\ increases. A better bound can be obtained by `binning' the heterogeneous collection into several homogeneous ensembles and utilizing the result of Corollary~\ref{cor:homEnsem}. 
\begin{cor} \label{cor:hetEnsem}
	For a heterogeneous collection of \numLoads\ loads with \numBins\ bins indexed by $\ell$ with \numLoadsBins\ loads in the $\ell^{th}$ bin, the ensemble \PSD\ is bounded as,
	\begin{align} \nonumber
	\sum_{\ell=1}^{\numBins}\numLoadsBins\sumOfPSDEll \leq S_{\bar{P}} \leq \numBins\sum_{\ell=1}^{\numBins}\numLoadsBins\sumOfPSDEll,
	\end{align}
	with
	\begin{align} \nonumber
		\sumOfPSDEll \triangleq \sum_{(i \ \in \ \text{bin} \ \ell)}S_{\tilde{P}^i}.
	\end{align}
\end{cor}
\begin{proof}
	Apply directly the same procedure as in Lemma~\ref{lem:psdBounds}, except for \numBins\ loads with the $\ell^{th}$ \PSD\ being  \numLoads\sumOfPSDEll\ (Corollary~\ref{cor:homEnsem}).
\end{proof}
%The bounds in Corollary~\ref{cor:hetEnsem} scale with \numBins, where \numBins $<<$ \numLoads, so that the bounds in Corollary~\ref{cor:hetEnsem} are tighter than the bounds in Lemma~\ref{lem:psdBounds}.
\subsection{Ensemble constraint set}
We now develop a constraint set on \sumOfPSDEll\, based on the set of constraints for the individual \PSD s that are in the definition of \sumOfPSDEll. In light of the results of Corollary~\ref{cor:hetEnsem}, we assume a homogeneous collection of loads. The idea is to sum the individual constraints over \numLoadsBins\ so that the definition of \sumOfPSDEll\ can be inserted. We do this for the rate constraint below, 
\begin{align} \nonumber
	&\sum_{(i \ \in \ \text{bin} \ \ell)}\int_{0}^{\infty}\left(1-\cos(\omega\delta)\right)S_{\tilde{P}^i}(\omega)d\omega \leq \numLoadsBins\frac{\pi\epsilon_2c^2_{2}}{2}, \\
	&\iff \nonumber \int_{0}^{\infty}\left(1-\cos(\omega\delta)\right)\sumOfPSDEll(\omega) d\omega \leq \numLoadsBins\frac{\pi\epsilon_2c^2_{2}}{2}.
\end{align}
The other constraints are obtained in a similar fashion, and the full constraint set for \sumOfPSDEll\ is: $\mathcal{S}(\numBins,\numLoadsBins,\ell) \triangleq$
\begin{align}
	\bigg\{ &\sumOfPSDEll \ : \ \nonumber \sumOfPSDEll \geq 0,  \\
	& \nonumber \int_{0}^{\infty} \sumOfPSDEll(\omega) d\omega\leq \numBins\numLoadsBins\bigg(\pi\epsilon_1 c_{1}^2\bigg), \nonumber\\
	&\int_{0}^{\infty} (1-\cos{(\omega\delta)}) \sumOfPSDEll(\omega) d\omega \leq \numBins\numLoadsBins\bigg(\frac{\pi\epsilon_2 c_{2}^2}{2}\bigg), \nonumber \\
	&\int_{0}^{\infty} \left|G(j\omega)\right|^2 \sumOfPSDEll(\omega)d\omega \leq \numBins\numLoadsBins\bigg(\pi\epsilon_3 c_{3}^2\bigg), \nonumber \\
	&\int_{0}^{\infty}\left|H(j\omega)\right|^2 \sumOfPSDEll(\omega)d\omega \leq \numBins\numLoadsBins\bigg( \pi \epsilon_4 c_{4}^2\bigg) \nonumber \bigg\}.
\end{align}
The sets $\mathcal{S}(1,\numLoadsBins,\ell)$ and $\mathcal{S}(\numBins,\numLoadsBins,\ell)$ represent the constraint set for the $\ell^{th}$ \PSD\ in the lower and upper bounds in Corollary~\ref{cor:hetEnsem}, respectively. Additionally, $\mathcal{S}(1,\numLoads,1)$ exactly represents the constraint set for the single ensemble \PSD\ in Corollary~\ref{cor:homEnsem} of a homogeneous ensemble with \numLoads\ loads.

\section{Resource Allocation for Flexible Loads} \label{sec:psdFeas}
We illustrate here how the ensemble constraint set in Section~\ref{sec:aggCap} can be used by a Balancing Authority (BA) for resource allocation. Conceptually, our proposed method ``projects'' the needs of the BA onto the ensemble constraint set. We first describe how a BA can incorporate the ensemble constraint set into an optimization problem for resource allocation and then how a BA can determine its needs spectrally.

\subsection{Allocation through projection}
{We denote by $S^{\loadDem}(\omega)$ a \PSD\ that represents the BA's needs. Computation of this \PSD\ will be discussed in the next section. Resource allocation is performed by projecting the \PSD\  $S^{\loadDem}(\omega)$ - what the grid needs - onto the Cartesian product $\times_{\ell=1}^{\numBins}\mathcal{S}(\numBins,\numLoadsBins,\ell)$ - what the ensemble of loads can provide. } The projection problem for a given value of $\numBins$ is,
\begin{align} \label{prob:optProb}
\min_{\{\sumOfPSDEll\}_1^{\numBins}} \quad &\int_{0}^{\infty}\left(\sumOfPSDAgg(\omega)  -S^{\loadDem}(\omega)\right)^2d\omega  \\ \nonumber
& \text{s.t.} \quad \forall \ \ell \in \{1,\dots, \numBins\}, \quad \sumOfPSDEll \in \mathcal{S}(\numBins,\numLoadsBins,\ell), \\ \label{eq:aggPSDCon}
&\qquad \sumOfPSDAgg = \sum_{\ell=1}^{\numBins}\sumOfPSDEll 
%\quad \forall \ \ell \in \{1,\dots, \mathscr{L}\}, \quad S^\ell_{\tilde{P}} \in \mathcal{S}^\ell.%, \\
%&\forall \ \omega \in [0,\infty), \quad \label{eq:noCapExceed} \sum_{\ell=1}^{\mathscr{L}}S^\ell_{\tilde{P}}(\omega) \leq S^{L}(\omega).
\end{align}
{Doing so allocates the needs of the grid across all loads; the loads will cover as much of the needs of the BA as they can while maintaining their QoS. In other words, the projection computes the regions shown in Figure~\ref{fig:Resource_Power} corresponding to each class of flexible loads.}
\begin{comment} \label{com:solveProbParam}
	If the ensemble of loads is heterogeneous, the BA can solve~\eqref{prob:optProb} twice. If the BA uses $\hat{n}$ bins, the first time the BA will set $\numBins = \hat{n}$ (upper bound) and the second time using $\numBins = 1$ (lower bound). Consequently, for a purely homogeneous collection the BA only needs to solve~\eqref{prob:optProb} once with $\numBins=1$, and \numLoads$_1$ = \numLoads\ (see Corollary~\ref{cor:homEnsem}). The LHS of~\eqref{eq:aggPSDCon} evaluated at the optimal solution of~\eqref{prob:optProb} then represents a bound on $S_{\bar{P}}$ for a heterogeneous ensemble, or is exactly $S_{\bar{P}}$ for a homogeneous ensemble.
\end{comment}
%The last constraint~\eqref{eq:noCapExceed} is included to ensure that the capacity is not over scheduled i.e., no class of loads should contribute more than the BA requires. 

%\begin{comment}
%In order to implement~\eqref{prob:optProb} on a computer, two steps are taken: (i) the problem is converted from continuous to discrete time and (ii) the resulting discrete time integrals in the objective and constraints are approximated with the trapezoidal integration formula. %The first step converts the region of integration from $[0,\infty)$ to $[0,\pi]$. The second step casts the problem to a finite dimensional convex optimization problem. One key requirement for the BA to solve~\eqref{prob:optProb} is to obtain a discrete time version of $S^L(\omega)$. We next describe a procedure to do this.	
%\end{comment}

\subsection{Spectral Needs of the BA} \label{sec:gridNeed}

In the following we provide an example procedure for a BA to spectrally determine its needs, as illustrated in Figure~\ref{fig:Resource_Power}. The BA first estimates the \PSD, $\Phi^{\netDem}$, of its net demand, i.e., demand minus renewable generation. It can estimate this quantity from time series data of demand and renewable generation, or through a modeling effort, or a combination thereof. The next step for the BA is to fit a parameterized model to $\Phi^{\netDem}$, which is termed $S^{\netDem}$. All controllable resources, including generators, flywheels, batteries, and flexible loads, together have to supply $\Phi^{\netDem}$ (or its parameterized model $S^{\netDem}$). The third step is  obtain the portion of $S^{\netDem}$ that flexible loads have to provide (similar to what is shown in Figure~\ref{fig:Resource_Power}) by ``filtering'' $S^{\netDem}$. Letting $F(j\omega)$ be an appropriate filter, then the \PSD\ of the signal the grid authority would like flexible loads to contribute is:
\begin{align}\label{eq:BPfilter}
S^{\loadDem}(\omega) = |F(j\omega)|^2 S^{\netDem}(\omega).
\end{align}

In the numerical example in this paper, we empirically estimate $\Phi^{\netDem}$ from time series data (from BPA, a balancing authority in the Pacific Northwest), and then obtain $S^{\netDem}$ by fitting an ARMA$(p,q)$ model to $\Phi^{\netDem}$. An example of $S^{\netDem}$ is shown as the orange line in Figure~\ref{fig:Resource_Power}, with $S^{\loadDem}$ being any of the shaded regions in Figure~\ref{fig:Resource_Power}. 
\begin{comment}
We have introduced a procedure for the BA to determine its needs in the spectral density domain. This procedure is \textit{completely} independent from our characterization of capacity presented. The next step is to use the results of the procedure, the BA's spectral needs, to find the closest \PSD\ of the loads to the grid's need.
\end{comment}

\section{Numerical Examples} \label{sec:numExp}
An example of determining the capacity of a collection of HVAC systems in commercial buildings is illustrated in this section. For the numerical experiments we construct homogeneous and heterogeneous ensembles from two types of HVAC systems. The parameters of each type are displayed in Table~\ref{tab:simParams_bldg}, and are chosen so that the HVAC systems are representative of those in small and large commercial buildings (hence the large and small superscripts in Table~\ref{tab:simParams_bldg}).

To aid exposition of the results, we define the following power and energy capacity indices,
\begin{align} \label{eq:aggCapInd}
\zeta^{\text{P}} &=  \frac{\textbf{Pow}(\sumOfPSDAgg)}{\textbf{Pow}(S^{\loadDem})} \times 100\%, \\
\zeta^{\text{E}} &=  \frac{\textbf{Eng}(\sumOfPSDAgg)}{\textbf{Eng}(S^{\loadDem})} \times 100\%,
\end{align}
so as to show the percentage of power and energy capacity required by the BA that can be covered by the loads. In the above, the numerator \PSD\ abstractly represents the LHS of~\eqref{eq:aggPSDCon} at the optimal solution of problem~\eqref{prob:optProb}. In all scenarios we solve the discrete time finite dimensional version of the convex optimization problem~\eqref{prob:optProb} with CVX~\cite{cvx}. That is, to obtain a finite dimensional optimization problem we convert from continuous to discrete frequency and approximate the resulting integrals on $[0,\pi]$\pb{or did you mean -$\pi$? Since the continuous version is from 0 I imagined you want to use 0 here?} with trapezoidal integration. All relevant simulation parameters, if not specified otherwise, can be found in Table~\ref{tab:simParams_bldg}.

\begin{table}
	\centering
	\def\arraystretch{1.4}%
	\caption{Simulation parameters}
	\label{tab:simParams_bldg}
	\begin{tabular}{|| l |c| c|| c |c| c||}
		Par.    & Unit      & Value  & Par.          & Unit         & Value \\
		\hline
		$c_{1}^{\textSmall}$   & kW        & 4      & $\gamma^{\textSmall}$  & $1/$hour     & 2.78\\
		$c_{2}^{\textSmall}$   & kW        & 0.8    & $\beta^{\textSmall}$   &$^\circ C/$kWh& 0.3597 \\
		$c_{4}^{\textSmall}, c_{4}^{\textLarge}$ &$^{\circ}C$& 1.11   & $\gamma^{\textLarge}$  & $1/$hour     & 177.6\\
		$c_{3}^{\textSmall}$   & kWh       & 0.5    & $\beta^{\textLarge}$   &$^\circ C/$kWh& 0.0450 \\
		$c_{1}^{\textLarge}$   & kW        & 40     & $T$           & Day        & 1 \\
		$c_{2}^{\textLarge}$   & kW        & 8      & $\{\epsilon_i\}_{i=1}^4$&N/A  & 0.05 \\
		$c_{3}^{\textLarge}$   & kWh       & 5      &$\delta$    & Sec     & 10  \\
		%$\delta$    & Sec     & 10       & $\{\epsilon_i\}_{i=1}^4$&N/A  & 0.05  \\
	\end{tabular}
\end{table}
\subsection{BA's spectral needs}\label{sec:obtain_Ref_PSD} 
The net demand data is collected from BPA (a BA in the pacific northwest United States). The \PSD\ of the net demand is determined using the method described in Section~\ref{sec:gridNeed}. The empirical net demand \PSD\ $\Phi^{\netDem}$ is estimated by the \emph{pwelch} command in MATLAB. We choose an ARMA(2,1) model to fit to the empirically estimated \PSD. Since the estimate $\Phi^{\netDem}$ will cap out at the Nyquist frequency $1/10$min, we extrapolate the net demand \PSD\ to the higher frequencies. The empirical \PSD\ (denoted $\Phi^{\netDem}$) and the extrapolated net demand \PSD\ (denoted $S^{\netDem}$) are shown in Figure~\ref{fig:NT_PSD}. 

We then choose two passbands to filter $S^{\netDem}$: (i) a low passband [$1/8$,$1/2$] (1/hour) and (ii)  a high passband [$1/30$,$1$] (1/min). The results of ``filtering''  $S^{\netDem}$ (see eq.~\eqref{eq:BPfilter}) are also shown in Figure~\ref{fig:NT_PSD}. The low passband \PSD\ is termed $S^{\loadDem}_{\textLow}$ and roughly corresponds to the region for TCLs in Figure~\ref{fig:Resource_Power}. The high passband \PSD\ is termed $S^{\loadDem}_{\textHigh}$ and roughly corresponds to the region for HVAC systems in Figure~\ref{fig:Resource_Power}.

\begin{figure}
	\centering
	\includegraphics[width=\columnwidth]{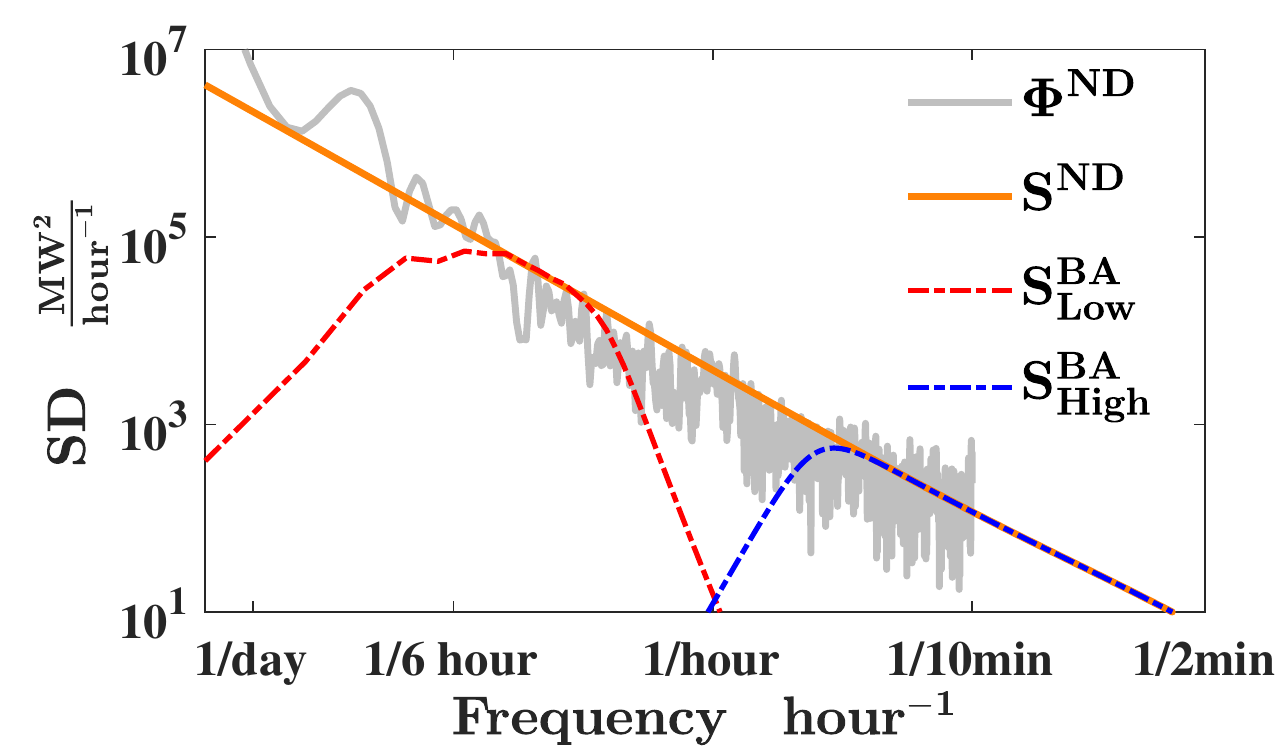}
	\caption{Empirical net demand \PSD\ , modeled \PSD\ for BPA's net demand, and the two reference \PSD's for the high and low frequency passband.}
	\label{fig:NT_PSD}
\end{figure}
\subsection{Meeting the BA's needs with large buildings} \label{subsec:resHomoLoad}
In this scenario we consider a homogeneous collection of large commercial buildings. The idea is to illustrate how many of these large commercial buildings would be required to meet the grids needs as a function of frequency. To do this, the two reference \PSD s obtained from the previous section are projected (by solving~\eqref{prob:optProb}) onto the \emph{same} ensemble constraint set with a varying number of large commercial buildings.

First, we show the results for $\numLoads = 15000$ large commercial buildings in Figure~\ref{fig:NT_PSD}. For the reference \PSD\ with high passband, the loads are able to meet the requirements of the grid, with an aggregate power and energy capacity index of $\zeta^{P} \approx 100 \%$ and $\zeta^{E} \approx 100 \%$. However, for the reference \PSD\ with low passband the loads are unable to meet the grids needs, with an aggregate power and energy capacity index of $\zeta^{P} = 43 \%$ and $\zeta^{E} = 26 \%$. 

For a varying number of loads the power and energy capacity index are recorded and plotted as a function of \numLoads\ in Figure~\ref{fig:homoPowEngInd}. These results indicate that the BA would require $\approx$ 5000 large commercial buildings in the higher passband and $\approx$ 32000 large commercial buildings for the low passband, to fulfill its needs. We also see that the BAs energy capacity requirement is met with fewer loads\pb{Dear co-authors, this is a common gramar error, please avoid it. You can have less milk, but you can only have fewer donuts. Not less donuts. When discussing countable things, use ``fewer''. When discussing uncountable things like milk, use ``less''.} (compared to its power capacity requirement) at the higher frequency range. In the lower frequency range, this conclusion is reversed.  

In summary, this numerical experiment suggests that the commercial HVAC systems considered here are more suitable to assist the grid by tracking signals with \PSD\ in the higher passband. If a BA wanted the commercial HVAC systems to track reference signals in the lower frequency range, than significantly more large commercial HVAC systems would be required.    
 
\begin{figure}
	\centering
	\includegraphics[width=1\columnwidth]{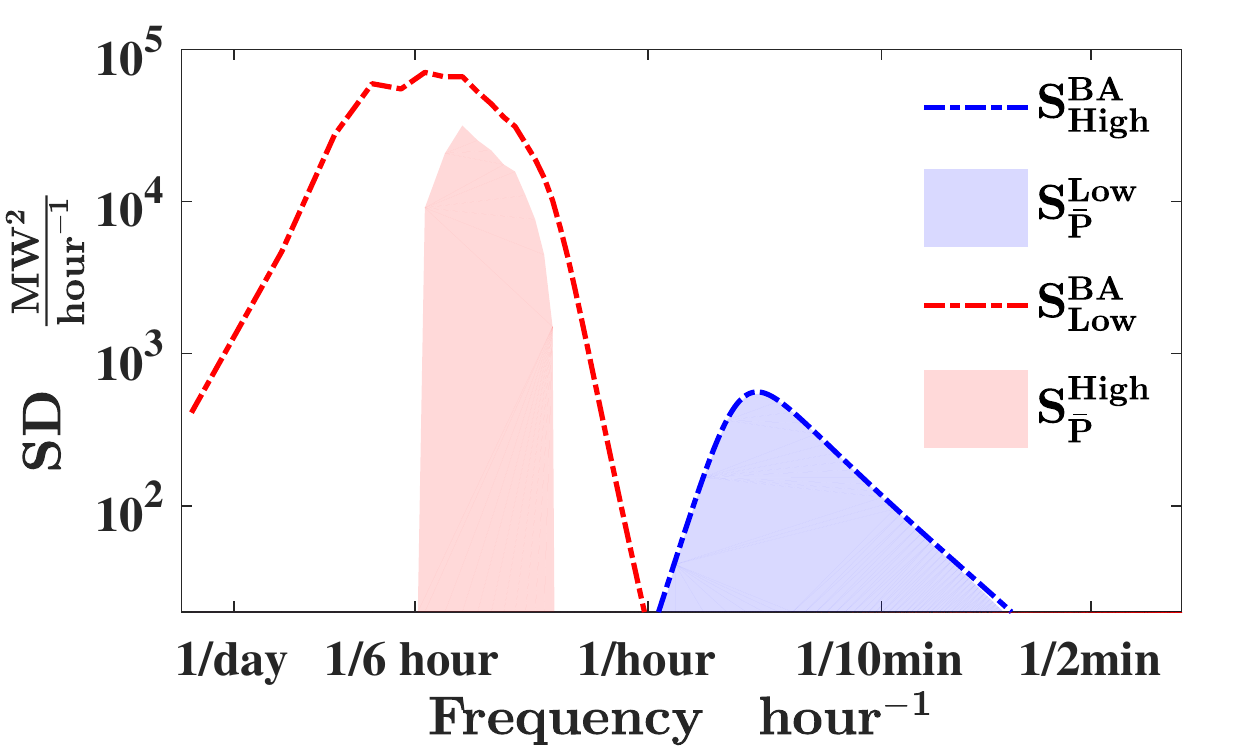}
	\caption{The two reference \PSD s and the ensembles \PSD\ (boundary of the shaded regions) obtained by solving~\eqref{prob:optProb} for a homogeneous collection of $\numLoads=15000$ loads.}
	\label{fig:Ref_PSD_case1}
	
	\centering
	\includegraphics[width=\columnwidth]{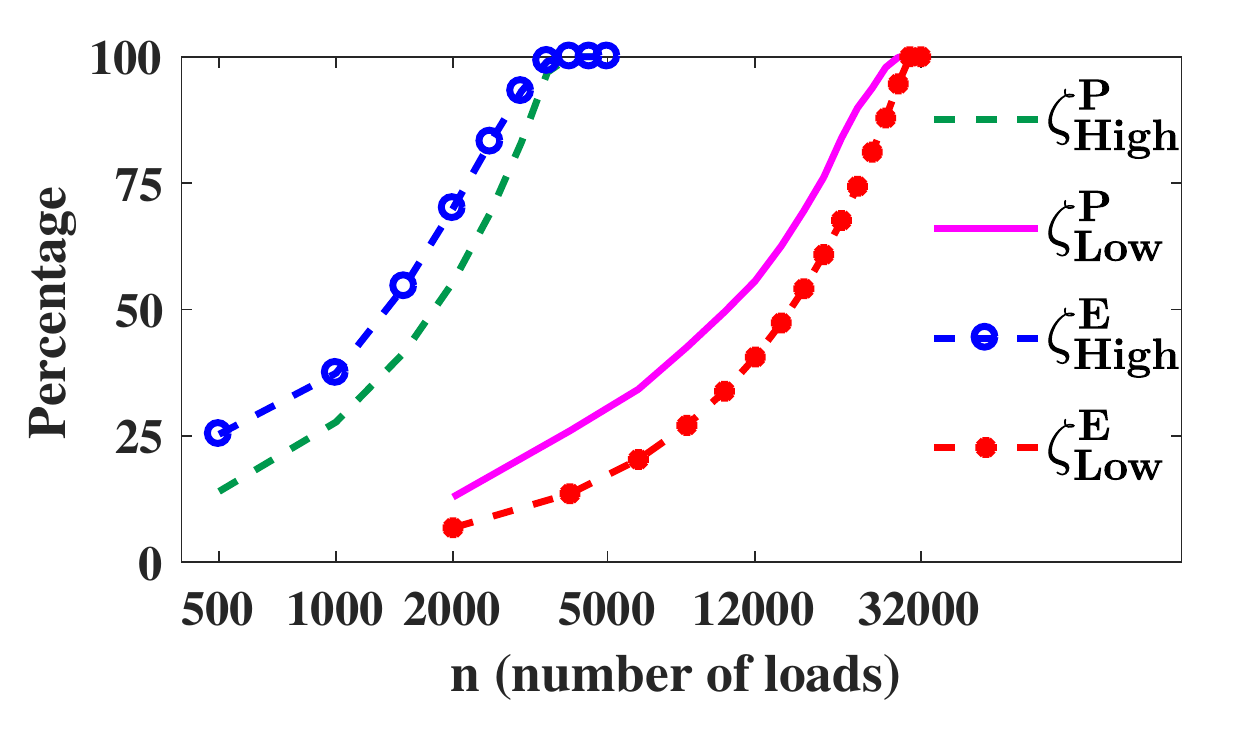}
	\caption{The aggregate power and energy capacity index for the two different reference \PSD s with a homogeneous collection of loads plotted against \numLoads.}
	\label{fig:homoPowEngInd}
	\centering
	
	\includegraphics[width=1\columnwidth]{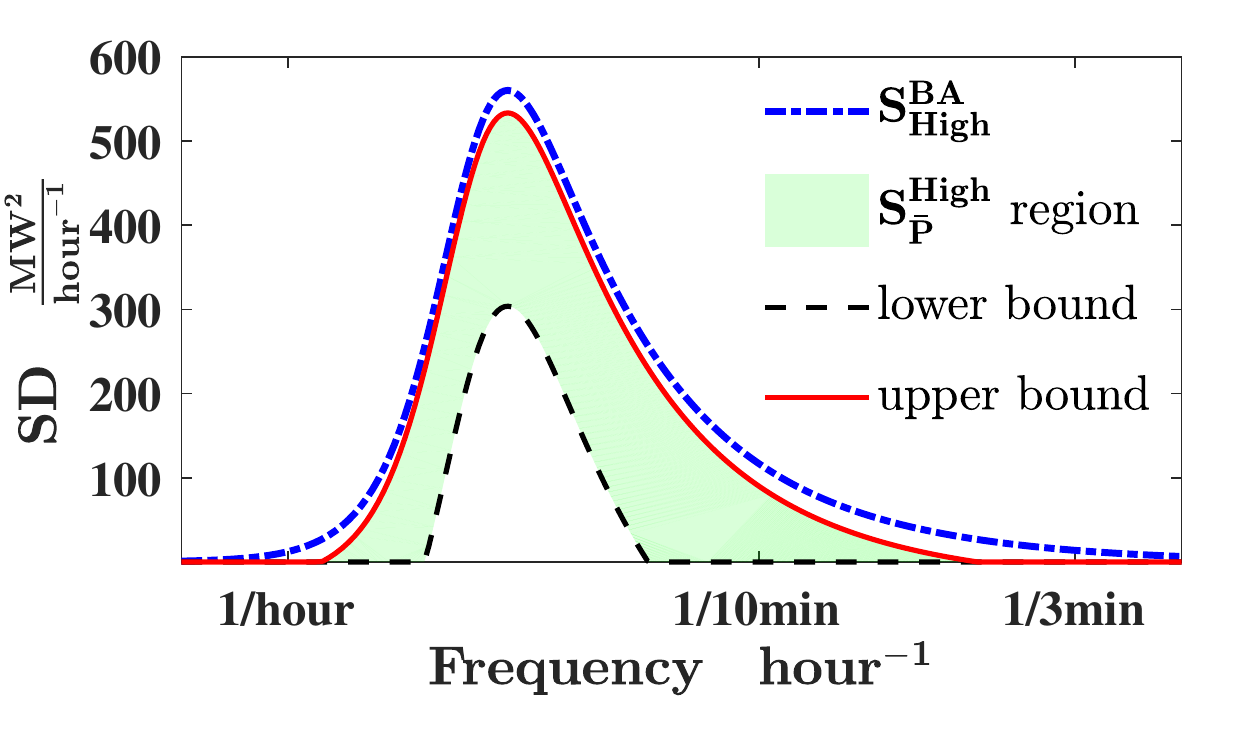}
	\caption{The upper and lower bounds on the ensemble \PSD\ for a heterogeneous ensemble.}
	\label{fig:hetEnsemBound}
\end{figure}
\subsection{Effects of heterogeneity}
We present a numerical experiment to illustrate the bounds in Corollary~\ref{cor:hetEnsem}. We consider $\numBins=2$ bins with, $\numLoads_1 = \numLoads_S = 900$ small and $\numLoads_2 = \numLoads_L = 2100$ large commercial buildings. We solve the problem~\eqref{prob:optProb} twice, with the appropriate parameters detailed in Comment~\ref{com:solveProbParam} for both times, so to obtain the upper and lower bounds in Corollary~\ref{cor:hetEnsem}. \ifshowArxiv This is done for the high passband in section~\ref{subsec:resHomoLoad} and a low passband of $[\frac{1}{3} \text{1/hours}, \frac{1}{15} \text{1/mins}]$. \fi \ifshowArxivAlt The \PSD\ that represents the grids needs in this experiment is the same \PSD\ with high passband for the experiments in section~\ref{subsec:resHomoLoad}. \fi

\ifshowArxivAlt
The results of this are shown in Figure~\ref{fig:hetEnsemBound}, where `upper bound' and `lower bound' represent the upper and lower bounds in Corollary~\ref{cor:hetEnsem}, respectively. Additionally, the shaded region in Figure~\ref{fig:hetEnsemBound} represents the region where the true ensemble \PSD\ lives. For this example, this gives us that the power capacity index is bounded as $17\% \leq \zeta^{P} \leq 69\%$ and the energy capacity index is bounded as $30\% \leq \zeta^{E} \leq 86\%$. 
\fi

\ifshowArxiv
The results of the high passband are shown in Figure~\ref{fig:hetEnsemBound}, where `upper bound' and `lower bound' represent the upper and lower bounds in Corollary~\ref{cor:hetEnsem}, respectively. Additionally, the shaded region in Figure~\ref{fig:hetEnsemBound} represents the region where the true ensemble \PSD\ lives. For this example, this gives us that the power capacity index is bounded as $17\% \leq \zeta^{P} \leq 69\%$ and the energy capacity index is bounded as $30\% \leq \zeta^{E} \leq 86\%$. 

The results of the low passband are shown in Figure~\ref{fig:hetEnsemBoundLowPass}, where the outline of the shaded green regions represent the bounds in Corollary~\ref{cor:hetEnsem}. Interestingly, the small and large buildings separate on the frequency axis. The large buildings cover the lower frequency components, whereas the small buildings cover the higher frequency components. 
\begin{figure}
	\includegraphics[width=1\columnwidth]{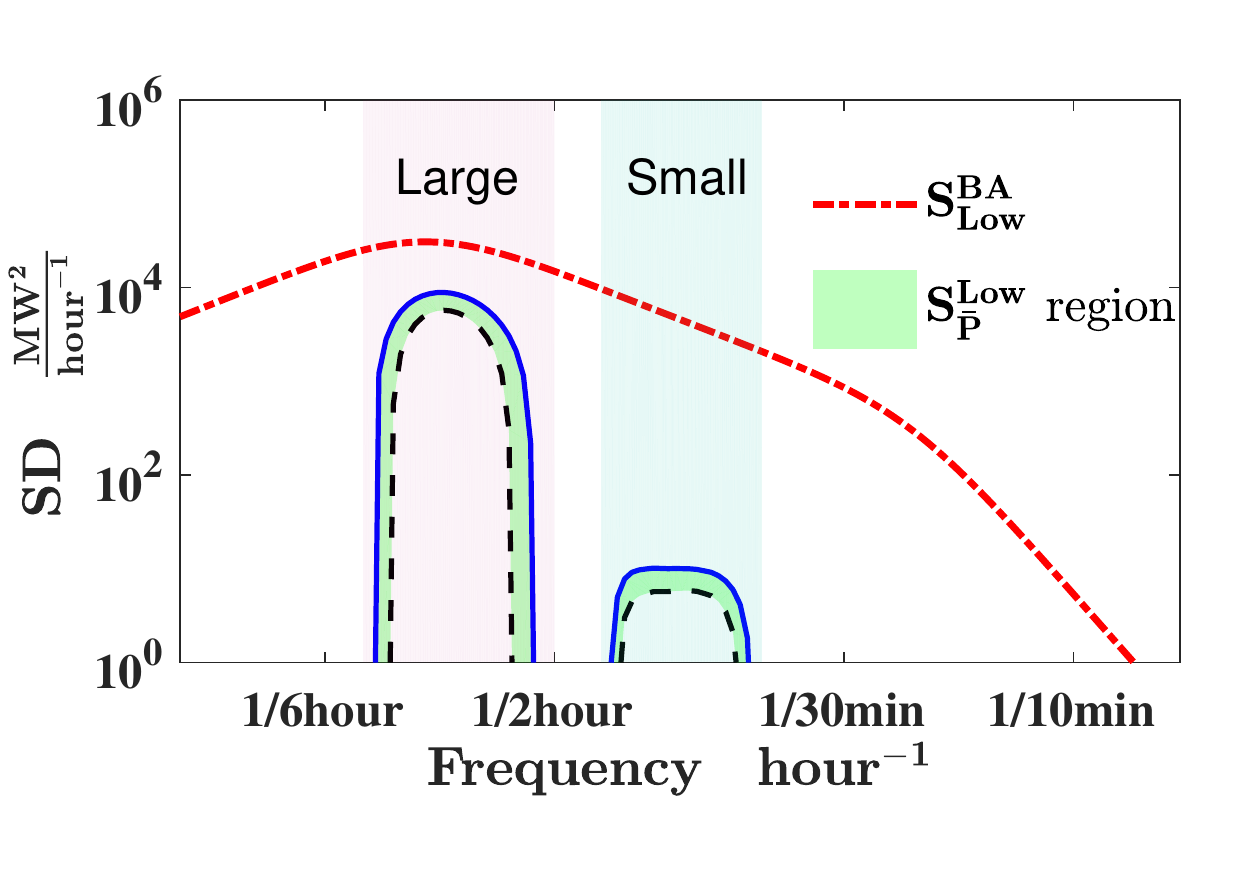}
	\caption{The upper and lower bounds on the ensemble \PSD\ for a heterogeneous ensemble with low passband.}
	\label{fig:hetEnsemBoundLowPass}	
\end{figure}
\fi

\section{Summary and conclusion}
Our characterization of virtual energy storage capacity of flexible loads allows a balancing authority to quantify which loads can contribute how much to mitigating supply demand mismatch. In contrast to past work, our method can be used for long term planning. The key insight is to characterize capacity through statistics of the reference signal rather than on specific instances of the signal. This framework also has another benefit: it enables us to define power and energy capacity of virtual energy storage in language of real energy storage: kW's and kWh's.  

An open problem is to extend the framework to non-stationary statistics of the grid's demand-supply mismatch and/or loads' demand deviation. Another open problem is the extension to the case of on/off loads. A key challenge for on/off loads will be to characterize their cycling constraint. Another avenue for future work is to consider a time of life QoS of the flexible load.

\ifx 0
\section{Appendix}
Suppose that we define the energy deviation as,
\begin{align} \label{eq:altEngDev}
	\tilde{E}(t) \triangleq \int_{t-T}^{t}\tilde{P}(\sigma)d\sigma,
\end{align}
with QoS constraint
\begin{align} \label{eq:altEngQoS}
	\left|\tilde{E}(t)\right| \leq c_4.
\end{align}
The transfer function relating $\tilde{P}$ to $\tilde{E}$ for the dynamics in~\eqref{eq:altEngDev} is,
\begin{align} \label{eq:altEngTF}
	G(s) = \frac{1 - e^{-sT}}{s},
\end{align}
with frequency response,
\begin{align} \label{eq:freqRespAlt}
	\left|G(j\omega)\right|^2 = 2\frac{\bigg(1-\cos(T\omega)\bigg)}{\omega^2}.
\end{align}
The transfer function~\eqref{eq:altEngTF} is BIBO stable so that from Proposition~\ref{prop:hajekTwo} the following holds,
\begin{align}
	S_{\tilde{E}}(\omega) = \left|G(j\omega)\right|^2S_{\tilde{P}}(\omega),
\end{align}
where $S_{\tilde{E}}$ is the \PSD\ of $\tilde{E}$.
Now the probabilistic form of~\eqref{eq:altEngQoS} is
\begin{align}
\mathcal{P}\bigg(\left|\tilde{E}(t)\right| \geq c_4\bigg) \leq \epsilon_4,
\end{align}
which through the Chebyshev inequality~\eqref{eq:chebIneq} it is sufficient to require,
\begin{align}
	\sigma_{\tilde{E}}^2 \leq c_4^2\epsilon_4,
\end{align}
with representation through the Wiener-Khinchin theorem,
\begin{align} \label{eq:wktAltRep}
	\frac{1}{2\pi}\int_{-\infty}^{\infty}\left|G(j\omega)\right|^2S_{\tilde{P}}(\omega)d\omega \leq c_4^2\epsilon_4.
\end{align}
Now substituting the quantity~\eqref{eq:freqRespAlt} into~\eqref{eq:wktAltRep} yields,
\begin{align}
	\int_{-\infty}^{\infty}\frac{\big(1-\cos(T\omega)\big)}{\omega^2}S_{\tilde{P}}(\omega)d\omega \leq \pi c_4^2\epsilon_4,
\end{align}
which is the exact same constraint as~\eqref{eq:intEngDev}.
\fi

\ifx 0
\section{Appendix}
%\balance
\subsection{Proof of Lemma~\ref{lem:setFeas}}
	\ifx 0
	Immediately $S_{\tilde{P}}(\omega) = 0 \ \forall \omega$ is in $\mathscr{S}$. Although, it is not hard to find another element of this set. Consider the PSD,
	\begin{align}
	S_{\tilde{P}}(\omega) = \begin{cases}
	\alpha, & \omega = \omega_0. \\
	0, & \text{otherwise},
	\end{cases}
	\end{align}  
	with $\alpha \in \Re^+$. Inevitably $\forall \omega_0 \ \exists \alpha $ such that this PSD is in $\mathscr{S}$. The uniform $\alpha$, over $\omega_0$, being $\alpha = 0$, which results in the previous example. Yet another example, 
	\fi
To show the set $\mathcal{S}$ is non empty consider the PSD,
	\begin{align} \nonumber
	S_{\tilde{P}}(\omega) = \begin{cases}
	0, & \omega = 0. \\
	0, & \omega \geq \omega_{H}. \\
	\alpha, & \text{otherwise}.
	\end{cases}	
	\end{align}
	with $\alpha \in \mathbb{R}^+$, which is in $\mathcal{S}$.	
	To show convexity, we break the set $\mathcal{S}$ into the intersection of two sets, one containing the integral constraints, $\mathcal{S}_1$, and one containing the constraint $S_{\tilde{P}}(\omega) \geq 0$. Since the set of positive real numbers is closed under addition, $S_{\tilde{P}}(\omega) \geq 0$ is convex. For the other set,
	it suffices to show that if $S_1 \in \mathcal{S}_1$ and $S_2 \in \mathcal{S}_1$, then for all $\alpha \in (0,1)$ $\alpha S_1 + (1-\alpha)S_2 \in \mathcal{S}_1$. 
	\ifx 0
	We rewrite the constraint $S_{\tilde{P}}(0) = 0$ in integral form as,
	\begin{align} \nonumber
	S_{\tilde{P}}(0) = \int_{-\omega_H}^{\omega_H}\delta(\omega)S_{\tilde{P}}(\omega)d\omega = 0,
	\end{align}
	where $\delta(\omega)$ is the standard generalized delta function. 
	\fi
	Now for the general integral constraint if,
	\begin{align} \nonumber 
	\int_{a}^{b}fS_1dx \leq \beta, \ \text{and} \ \int_{a}^{b}fS_2dx \leq \beta
	\end{align}
	then,
	\begin{align} \nonumber
	\int_{a}^{b}f[\alpha S_1 + (1-\alpha)S_2]dx \leq \alpha \beta + (1-\alpha)\beta = \beta
	\end{align} 
	which proves the result, as $\mathcal{S}$ is the intersection of a finite number of convex sets $\square$.
\fi
\end{document}